\documentclass[12pt,reqno]{amsart}
\usepackage{graphicx}
\usepackage{tikz-cd}
\usepackage{amsmath}
\usepackage{amssymb}
\usepackage{amscd}
\usepackage{mathrsfs}
\usepackage[all,cmtip]{xy}
 \usepackage{color}
\usepackage{amsthm}
\usepackage{caption}

\theoremstyle{definition}

\theoremstyle{definition}

\theoremstyle{definition}
\newtheorem{remark}{Remark}
\theoremstyle{plain}
\newtheorem{theorem}{Theorem}

\theoremstyle{plain}

\newcommand\com[1]{}

\newcommand\Cc{{\let\mathcal\mathscr\mathcal C}}

\newcommand\D{{\mathcal D}}

\newcommand\E{\mathcal{E}}

\newcommand\op[1]{\mathop{\rm #1}\nolimits}

\newcommand\p{\partial}

\newcommand\R{{\mathbb R}}

\newcommand\ric{\mathcal{R}}

\newcommand{\sym}{\mathfrak{g}}
\newcommand{\Sym}{\mathcal{G}}

\def\presuper#1#2%
{\mathop{}%
	\mathopen{\vphantom{#2}}^{#1}%
	\kern-\scriptspace%
	#2}

\begin{document}

 \title{Differential invariants of Kundt waves}
 \author[B. Kruglikov, D. McNutt, E. Schneider]{Boris Kruglikov$^{\dagger\ddagger}$, David McNutt$^\ddagger$, Eivind Schneider$^\dagger$}
 \date{}
\address{\hspace{-17pt}
$^\dagger$Department of Mathematics and Statistics, UiT the Arctic University of Norway, Troms\o\ 90-37, Norway.\newline
E-mails: {\tt boris.kruglikov@uit.no, eivind.schneider@uit.no}. \newline
$^\ddagger$Department of Mathematics and Natural Sciences, University of Stavanger, 40-36 Stavanger, Norway.\newline
E-mail: {\tt david.d.mcnutt@uis.no}. }
 \keywords{Lorentzian metric, scalar curvature invariant, Cartan invariant, differential invariant, invariant derivation, Poincar\'e function}

 \vspace{-14.5pt}
 \begin{abstract}
Kundt waves belong to the class of spacetimes which are not distinguished by their scalar curvature
invariants. We address the equivalence problem for the metrics in this class via scalar differential
invariants with respect to the equivalence pseudo-group of the problem.
We compute and finitely represent the algebra of those on the generic stratum and also specify the
behavior for vacuum Kundt waves. The results are then compared to the invariants computed
by the Cartan-Karlhede algorithm.
 \end{abstract}

 \maketitle

\section*{Introduction}\label{S0}

The Kundt waves can be written in local coordinates as follows
 \begin{equation}\label{KW}
g = dx^2+dy^2-du\,\Bigl(dv-\tfrac{2v}x\,dx+\bigl(8xh-\tfrac{v^2}{4x^2}\bigr)\,du\Bigr),
 \end{equation}
where $h=h(x,y,u)$ is an arbitrary function. In order for $g$ to be vacuum, $h$ must be harmonic in $x,y$.
These metrics were originally defined by Kundt \cite{Ku} in 1961, as a special class of pure radiation spacetimes of
Petrov type III or higher, admitting a non-twisting, non-expanding shear-free
null congruence $\ell$ \cite{ESEFE}:
$g(\ell,\ell)=0$, $\op{Tr}_g(\nabla\ell)=0$, $\|\nabla\ell\|_g^2=0$.

All Weyl curvature invariants \cite{W}, i.e.\ scalars constructed from tensor products of covariant
derivatives of the Riemann curvature tensor by complete contractions, vanish for these spacetimes. Thus,
these plane-fronted metrics belong to the collection of VSI spacetimes, where all polynomial
scalar curvature invariants vanish \cite{PPCM}. These spaces have been extensively explored in the
literature \cite{CHP,CHPP}.

Since it is impossible to distinguish Kundt waves from Minkowski spacetime by Weyl curvature invariants,
other methods have been applied. In \cite{MMC} Cartan invariants have been computed for vacuum
Kundt waves and the maximum iteration steps in Cartan-Karlhede algorithm was determined.
Cartan invariants allow to distinguish all metrics, 
but initially they are functions
on the Cartan bundle, also known as the orthonormal frame bundle, not on the original spacetime.

Cartan invariants are polynomials in structure functions of the canonical frame (Cartan connection)
and their derivatives along the frame \cite{C}. Thus they are obtained from the components of the
Riemann curvature tensor and its covariant derivatives without complete contractions.
Absolute invariants are chosen among those that are
invariant with respect to the structure group of the Cartan bundle.
This is usually achieved by a normalization of the group parameters \cite{C,O}.

When the frame is fixed (the structure group becomes trivial) the Cartan invariants descend
to the base of the Cartan bundle, i.e.\ the spacetime
(in some cases, which we do not consider, the frame cannot be completely fixed but then 
the form of the curvature tensor and its covariant derivatives are unaffected by the frame freedom).
The Cartan-Karlhede algorithm \cite{Kar,ESEFE} specifies when the normalization terminates
and how many derivatives of the curvature along the frame are involved in the final list of invariants.

In this paper we propose another approach, which originates from the works of Sophus Lie. Namely we
distinguish spacetimes by scalar differential invariants of their metrics. The setup is different:
we first determine the equivalence group of the problem that is the group preserving the class
of metrics under consideration. It is indeed infinite-dimensional and local, so it is more proper
to talk of a Lie pseudogroup, or its Lie algebra sheaf. Then we compute invariants of this pseudogroup
and its prolonged action. The invariants live on the base of the Cartan bundle, i.e.\ the spacetime,
but they are allowed to be rational rather than polynomial in jet-variables (derivatives of the
metric components). We recall the setup in Section \ref{S1}.

Recently \cite{KL2} it was established that the whole infinite-dimensional algebra of invariants
can be finitely generated in Lie-Tresse sense. This opens up an algebraic approach to the classification,
and that is what
we implement here. We compute explicitly the generating differential invariants and invariant derivations,
organize their count in Poincar\'e series, and resolve the equivalence problem for generic metrics within the class.
We also specify how this restricts to vacuum Kundt waves. This is done in Sections \ref{S2}-\ref{S3}.
More singular spaces can be treated in a manner analogous to our computations.

Since vacuum Kundt waves have already been investigated via the Cartan method \cite{MMC}, we include
a discussion on the correspondence of the invariants in this case. 
This correspondence does not preserve the order of invariants, because the approaches differ,
and we include a general comparison of the two methods. This is done in Section \ref{S4}.

\section{Setup of the problem: actions and invariants}\label{S1}

Metrics of the form (\ref{KW}) are defined on an open subset of the manifold $M= (\R \setminus \{0\}) \times \R^3\subset\R^4$.
Thus a metric $g$ can be identified as a (local) section of the bundle
$\pi\colon M \times \R\to M$ with the coordinates $x,y,u,v,h$. We denote the total space of the bundle by $E$.
The Kundt waves then satisfy the condition $h_v=0$. This partial differential equation (PDE)
determines a hypersurface $\E_1$ in $J^1 \pi$.

Here $J^k\pi$ denotes the $k$-th order jet bundle. This space is diffeomorphic to $M \times \mathbb R^N$, where $N=\tbinom{k+4}{4}$,
and we will use the standard coordinates $h,h_x,h_y,...,h_{u v^{k-1}},h_{v^{k}}$ on $\mathbb R^N$.
Function $h=h(x,y,u,v)$ determines the section $j^kh$ of $J^k\pi$ in which those standard coordinates are the usual partial derivatives of $h$.

The space $J^k\pi$ comes equipped with a distribution (a sub-bundle of the tangent bundle),
called the Cartan distribution.
A PDE of order $k$ is considered as a submanifold of $J^k \pi$, and its solutions correspond to maximal integral manifolds of the Cartan distribution restricted to the PDE. For a detailed review of jets, we refer to \cite{O,KL1}.
The prolongation $\E_k\subset J^k\pi$ is the locus of differential corollaries of the defining equation
of $\E_1$ up to order $k$. We also let $\E_0=J^0\pi=E$.

The vanishing of the Ricci tensor is equivalent to the condition $h_{xx}+h_{yy}=0$.
This yields a sub-equation $\ric_2 \subset \E_2 \subset J^2\pi$, whose prolongations we denote by
$\ric_k \subset J^k \pi$. Since this case of vacuum Kundt waves was considered thoroughly in \cite{MMC}
we will focus here mostly on general Kundt waves. However, after finding the differential invariants in
the general case it is not difficult to describe the differential invariants in the vacuum case.
This will be done in Section \ref{S3}.

\subsection{Lie pseudogroup}\label{S1.1}

The Lie pseudogroup of transformations preserving the shape (i.e.\ form of the metric)
can be found by pulling back $g$ from (\ref{KW}) through a general transformation
$(\tilde x,\tilde y,\tilde u,\tilde v) \mapsto (x,y,u,v)$,
and then requiring that the obtained metric is of the same shape:
 \[
d \tilde x^2+d \tilde y^2-d \tilde u\,\Bigl(d \tilde v-\tfrac{2\tilde v}{\tilde x}\,d\tilde x+
\bigl(8\tilde x \tilde h-\tfrac{\tilde v^2}{4\tilde x^2}\bigr)\,d\tilde u \Bigr).
 \]
This requirement can be given in terms of differential equations on $x,y,u,v$ as functions of
$\tilde x,\tilde y,\tilde u,\tilde v$, with the (invertible) solutions described below.
The obtained differential equations are independent of whether the Kundt wave is Ricci-flat or not,
so the shape-preserving Lie pseudogroup is the same for both general and Ricci-flat Kundt waves.

A pseudogroup preserving shape (\ref{KW}) contains transformations of the form
(we also indicate their lift to $J^0\pi=E$)
 \begin{align}
x \mapsto x, \quad y &\mapsto y+C, \quad u \mapsto F(u), \quad v \mapsto \frac{v}{F'(u)}-2 \frac{F''(u)}{F'(u)^2} x^2,
\label{e2}\\
h &\mapsto \frac{h}{F'(u)^2}+\frac{2 F'''(u) F'(u)-3 F''(u)^2}{8 F'(u)^4}x,
\label{e3}
 \end{align}
where $F$ is a local diffeomorphism of the real line, i.e.\ $F'(u)\neq0$.
This Lie pseudogroup was already described in \cite{PPCM}, formula (A.37).

Transformations \eqref{e2}-\eqref{e3} form the Zariski connected component $\Sym_0$ of
the entire Lie pseudogroup $\Sym$ of shape-preserving transformations.
(Note that $\Sym_0$ differs from the topologically connected component of unity given by $F'(u)>0$.)
The pseudogroup $\Sym$ is generated, in addition to transformations \eqref{e2}-\eqref{e3}, by the maps
$y \mapsto - y$ and $(x,h) \mapsto (-x, -h)$ preserving shape (\ref{KW}).
Note that $\Sym/\Sym_0=\mathbb Z_2 \times \mathbb Z_2$.

The Lie algebra sheaf $\sym$ of vector fields corresponding to $\Sym$ (and $\Sym_0$) is spanned by the vector fields
\begin{equation}
X = \partial_y, \quad Y(f) = 4 f \partial_u-(4 v f'+8x^2 f'') \partial_v+(xf'''-8hf') \partial_h
\end{equation}
where $f=f(u) \in C_{\text{loc}}^\infty(\mathbb R)$ is an arbitrary function.

When looking for differential invariants, it is important to distinguish between $\Sym$ and $\Sym_0$.
Firstly, differential $\Sym_0$-invariants need not be $\Sym$-invariant. Secondly, a set of differential
invariants that separates $\Sym$-orbits as a rule will not separate $\Sym_0$-orbits.
We will restrict our attention to the $\Sym$-action while outlining the changes needed to be made
for the other choices of the Lie pseudogroup.

\subsection{Differential invariants and the global Lie-Tresse theorem}\label{S1.2}

A differential invariant of order $k$ is a function on $\E_k$ which is constant on orbits of $\Sym$.
In accordance with \cite{KL2} we consider only invariants that are rational in the fibers of $\pi_k:\E_k\to E$
for every $k$.

The global Lie-Tresse theorem states that for algebraic transitive Lie pseudogroups, rational differential invariants separate orbits in general position in $\E_\infty$ (i.e.\ orbits in the complement of a
Zariski-closed subset), and the field of rational differential invariants is generated by a finite number
of differential invariants and invariant derivations. In fact it suffices to consider the (sub)algebra
of invariants that are rational on fibers of $\pi_{\ell}:\E_\ell\to E$ and polynomial on fibers
of $\pi_{k,\ell}:\E_k\to\E_\ell$ for some $\ell$. In the case of Kundt waves we will show that $\ell=2$.
For simplicity we will mostly discuss the field of rational invariants in what follows.

We refer to \cite{KL2} for the details of the theory which holds for transitive Lie pseudogroups.
The Lie pseudogroup we consider is not transitive: the $\Sym$-orbit foliation of $E$ is $\{x=\op{const}\}$.
Let us justify validity of a version of the Lie-Tresse theorem for our Lie pseudogroup action.

For every $a\in E$ the action of the stabilizer of $a$ in $\Sym_0$ is algebraic on the fiber $\pi_{\infty,0}^{-1}(a)$,
and so for every $k$ and $a$ we have an algebraic action of a Lie group on the algebraic manifold of
$\pi_{k,0}^{-1}(a)$. By Rosenlicht's theorem rational invariants separate orbits in general position.
It is important that the dependence of the action on $a$ is algebraic.

From the description of the $\Sym_0$ action on $E$ it is clear that orbits in general position intersect with
the fiber over $a(x)=(x,0,0,0,1)$ for a unique $x\in\R\setminus\{0\}$. A $\Sym$-orbit in $\E_\infty$
intersecting with the fiber of $a(x)$ intersects $a(-x)$ as well. Thus we can separate orbits with scalar
differential invariants, in addition to the invariant $x$ or $x^2$, for $\Sym_0$ or $\Sym$ respectively.
It is not difficult to see, following \cite{KL2}, that in our case the field of differential
invariants is still finitely generated. We skip the details because this will be apparent from our explicit
description of the generators of this field in what follows.

\subsection{The Hilbert and Poincar\'e functions}\label{S1.3}

The transcendence degree of the field of rational differential invariants of order $k$
(that is the minimal number of generators of this field, possibly up to algebraic extensions)
is equal to the codimension of the $\sym$-orbits in general position in $\E_k$. The results in this section are valid for both $\Sym_0$ and $\Sym$ and all intermediate Lie pseudogroups (there are three of them since the quotient $\Sym/\Sym_0$ is the Klein four-group).

For $k\geq 0$, the dimension of $J^k \pi$ is given by
\[\dim J^k \pi = 4+ \binom{k+4}{4}.  \]
The number of independent equations defining $\E_k$ is $\binom{k+3}{4}$ which yields \[\dim \E_k = \dim J^k \pi- \binom{k+3}{4} = 4+ \binom{k+3}{3}, \quad k \geq 0.\]

For small $k$, the dimension of a $\sym$-orbit in $J^k \pi$ in general position may be found by computing the dimension of the span of $\sym|_{\theta_k} \subset T_{\theta_k} J^k\pi$ for a general point $\theta_k \in J^k\pi$. It turns out that
the equation $\E_k$ intersects with regular orbits, so we get the same results by choosing $\theta_k\in\E_k$.

\begin{theorem}\label{counting}
	The dimension of a $\sym$-orbit in general position in $\E_k$ is $4$ for $k=0$ and it is equal to $k+5$ for $k>0$.
\end{theorem}

\begin{proof}
We need to compute the dimension of the span of $X^{(k)}$ and $Y(f)^{(k)}$ at a point in general position in $\E_k$. The $k$-th prolongation of the vector field $Y(f)$ is given by
\begin{equation}
 Y(f)^{(k)}= 4 f \D_u^{(k+1)}-(4 v f'+8x^2 f'') \D_v^{(k+1)} + \sum_{|\sigma| \leq k} \D_\sigma (\phi) \partial_{h_\sigma} \label{prolongation}
\end{equation}
where $\sigma=(i_1,\dots,i_t)$ is a multi-index of length $|\sigma|=t$
($i_j$ corresponds to one of the base coordinates $x,y,u,v$),
$\D_\sigma=\D_{i_1}\cdots\D_{i_t}$ is the iterated total derivative, $\D_i^{k+1}$ is the truncated total derivative
as a derivation on $J^k\pi$, and
 \begin{align*}
\phi =\,& Y(f)\lrcorner\,(dh-h_xdx-h_ydy-h_udu-h_vdv)\\
=\,& x f'''-8hf'-4f\,h_u+(4vf'+8x^2 f'')\,h_v
 \end{align*}
is the generating function for $Y(f)$; we refer to Section 1.5 in \cite{KL1}.
We see that the $k$-th prolongation depends on $f,f',...,f^{(k+3)}$.

We can without loss of generality assume that the $u$-coordinate of our point in general position is $0$, since $\partial_u$ is contained in $\sym$. At $u=0$ the vector field $Y(f)^{(k)}$ depends only on the $(k+3)$-degree Taylor polynomial of $f$ at $u=0$, which implies that there are at most $k+4$ independent vector fields among these. Adding the vector field $X^{(k)}$
to them gives $k+5$ as an upper bound of the dimension of an orbit.

Let $\theta_k \in \E_k$ be the point defined by $x=1, h=1$, with all other jet-variables set to $0$ and let $Z_m=Y(u^m)$. It is clear from (\ref{prolongation}) that the $k$-th prolongations of $X, Z_0,...,Z_{k+3}$ span a $(k+5)$-dimensional subspace of $T_{\theta_k} \E_k$, implying that $k+5$ is also a lower bound for the dimension of an orbit in general position and verifying the claim of the theorem.
\end{proof}

Let $s_k^\E$ denote the codimension of an orbit in general position inside of $\E_k$, i.e.\ the number of independent
differential invariants of order $k$. It is given by
 \[
s^\E_0=1\ \text{ and }\ s_k^\E= \frac{k}{6}(k+5)(k+1) \text{ for } k \geq 1.
 \]
The Hilbert function $H_k^\E=s_k^\E-s_{k-1}^\E$ is given by
 $$
H_0^\E=H_1^\E=1\ \text{ and }\
H_k^\E = \frac{k(k+3)}{2} \text{ for } k\geq 2.
 $$
This counts the number of independent differential invariants of ``pure'' order $k$.
For small $k$ the results are summed up in the following table.
 \[
\begin{array}{|c|rrrrrrr|}
\hline
k	& 0 & 1 & 2 & 3 & 4 & 5 & 6\\ \hline
\dim J^k \pi & 5 & 9 & 19 & 39 & 74 & 130 & 214  \\
\dim \E_k	& 5 & 8 & 14 & 24 & 39 & 60 & 88 \\
\dim \mathcal{O}_k	& 4 & 6 & 7 & 8 & 9 & 10 & 11  \\
s_k^{\E} & 1 & 2 & 7 & 16 & 30 & 50 & 77 \\
H_k^{\E} & 1 & 1 & 5 & 9 & 14 & 20 & 27 \\\hline
\end{array}
 \]
The corresponding Poincar\'e function $P_\E(z)=\sum_{k=0}^\infty H_k^\E z^k$ is given by
 \[
P_\E(z) = \frac{1-2z+5z^2-4z^3+z^4}{(1-z)^3}.
 \]

\section{Differential invariants of Kundt waves}\label{S2}

We give a complete description of the field of rational differential invariants. We will focus
on the action of the entire Lie pseudogroup $\Sym$ (with four Zariski connected components),
while also describing what to do if one wants to consider only one (or two) connected components.

\subsection{Generators}\label{S2.1}

The second order differential invariants of the $\Sym$-action are generated by the following seven functions
 \begin{align*}
I_0 &= x^2, &
I_1 &= \frac{(x h_x-h)^2}{h_y^2}, &
I_{2a} &= \frac{h_{xx}}{x h_x-h}, \\
I_{2b} &= \frac{x h_{xy}}{h_y},   &
I_{2c} &= \frac{h_{yy}}{x h_x-h}, &
I_{2d} &= \frac{(x^2 h_{yu}-v h_y)^2}{x (x h_x-h)^3 }, \\
& \hphantom{a}\hskip20pt I_{2e} = \frac{(x^3 h_{xu}-v x h_x-x^2 h_u+v h)(x h_x-h)}{(x^2 h_{yu}-v h_y) h_y} \hspace{-7cm}
 \end{align*}
and these invariants separate orbits of general position in $\E_2$. They are independent as functions on $\E_2$,
and one verifies that the number of invariants agrees with the Hilbert function $H_k^\E$ for $k=0,1,2$.

Note that $\sqrt{I_0}=x$ and $\sqrt{I_1} = \tfrac{x h_x-h}{h_y}$ are not invariant under the discrete transformations $(x,h)\mapsto (-x,-h)$ and $y \mapsto -y$. They are however invariant under the
Zariski connected pseudogroup $\Sym_0$ and should be used for generating the field of differential
$\Sym_0$-invariants, since the invariants above do not separate $\Sym_0$-orbits on $\E_2$.

 \begin{remark}
If $\mathcal A_2$ denotes the field of second order differential $\Sym$-invariants and $\mathcal B_2$ the field of second order differential $\Sym_0$-invariants, then $\mathcal B_2$ is an algebraic field extension of $\mathcal A_2$ of degree $4$ and its Galois group is $\Sym/\Sym_0= \mathbb Z_2 \times \mathbb Z_2$. Intermediate pseudogroups lying between $\Sym_0$ and $\Sym$ are in one-to-one correspondence with subgroups of $\mathbb Z_2 \times \mathbb Z_2$ that, by Galois theory,  are
in one-to-one correspondence with algebraic field extensions of $\mathcal A_2$ that are
contained in $\mathcal B_2$.
	
Including $\mathcal B_2$ there are four such nontrivial algebraic extensions of $\mathcal A_2$, and they are the splitting fields of the polynomials $t^2-I_0$, $t^2-I_1$, $t^2-I_0 I_1$ and  $(t^2-I_0)(t^2-I_1)$ over $\mathcal A_2$, respectively.
	
Higher-order invariants are generated by second-order invariants and invariant derivations, so the field of all differential invariants
depends solely on the chosen field extension of $\mathcal A_2$.
 \end{remark}


In order to generate higher-order differential invariants we use invariant derivations, i.e.\ derivations on $\E_\infty$ commuting with the $\Sym$-action. It is not difficult to check that the following derivations are invariant.
 \begin{gather*}
\nabla_1 =x D_x+2v D_v,\quad
\nabla_2= \frac{x h_x-h}{h_y}\,D_y,\quad
\nabla_4=\frac{x^2 h_{yu}-v h_y}{h_y}\,D_v,\\
\nabla_3=\frac{h_y}{x^2 h_{yu}-v h_y} \left(D_u-\Bigl(8x^2 h_x-\frac{v^2}{4x^2}\Bigr) D_v  \right).
 \end{gather*}

 \begin{theorem}\label{generators}
The field of rational scalar differential invariants of $\Sym$ is generated by the second-order invariants
$I_0,I_1,I_{2a},I_{2b},I_{2c},I_{2d},I_{2e}$
together with the invariant derivations $\nabla_1,\nabla_2,\nabla_3,\nabla_4$.

The algebra of rational differential invariants, which are polynomial
starting from the jet-level $\ell=2$, over $\mathcal A_2$, $\mathcal B_2$ or an intermediate field,
depending on the choice of Lie pseudogroup, is generated by the above seven
second-order invariants (with possible passage from $I_0$ to $\sqrt{I_0}$ and from $I_1$ to $\sqrt{I_1}$)
and the above four invariant derivations.
 \end{theorem}

 \begin{proof}
We shall prove that the field generated by the indicated differential invariants and invariant derivations
for every $k>2$ contains $H_k^\E=\frac{k(k+3)}{2}$ functionally independent invariants, and moreover
that their symbols are quasilinear and independent. This together with the fact that the indicated invariants
generate all differential invariants of order $\leq2$ implies the statement of the theorem.

We demonstrate by induction in $k$ a more general claim that there are $H_k^\E$
quasilinear differential invariants of order $k$ with the symbols at generic $\theta_{k-1}\in J^{k-1}\pi$
proportional to $h_{x^i y^{j} u^l}$, where $i+j+l=k$ and $0 \leq l <k$. The number of such $k$-jets is indeed equal to
the value of the Hilbert function $H_k^\E$.

The base $k=3$ follows by direct computation of the symbols of $\nabla_1I_{2a}, \nabla_1I_{2b}, \nabla_1I_{2c},
\nabla_1I_{2d}, \nabla_1I_{2e}, \nabla_2I_{2c}, \nabla_2I_{2d}, \nabla_3I_{2d}, \nabla_3I_{2e}$.
Assuming the $k$-th claim, application of $\nabla_1$ gives $k(k+3)/2$ differential invariants of order $k+1$,
and $\nabla_2$ adds $k$ additional differential invariants, covering the symbols $h_{x^i y^{j} u^l}$
with $i+j+l=k+1$ and $0 \leq l <k$.
Further application of $\nabla_3$ gives $2$ more differential invariants with symbols $h_{xu^k}$, $h_{yu^k}$.
Thus the invariants are independent and the calculation
 \[
\frac{k(k+3)}{2} + k+2= \frac{(k+1)(k+4)}{2}
 \]
completes the induction step.

For the algebra of invariants it is enough to note that our generating set produces
invariants that are quasi-linear in jets of order $\ell=2$ or higher, and so any differential invariant
can be modified by elimination to an element in the base field $\mathcal A_2$, $\mathcal B_2$
or an intermediate field.
 \end{proof}

 \begin{remark}
As follows from the proof it suffices to have only derivations $\nabla_1,\nabla_2,\nabla_3$.
Yet $\nabla_4$ is obtained from those by commutators.
 \end{remark}

It is possible to give a more concise description of the field/algebra of differential invariants than that of
Theorem \ref{generators}. Let $\alpha_i$ denote the horizontal coframe dual to the derivations $\nabla_i$,
i.e.
 \begin{gather*}
\alpha_1 =\frac1x\,dx,\quad
\alpha_2= \frac{h_y}{x h_x-h}\,dy,\quad
\alpha_3=\frac{x^2 h_{yu}-v h_y}{h_y}\,du,\\
\alpha_4=\frac{h_y}{x^2 h_{yu}-v h_y} \left(dv-\frac{2v}x\,dx+\Bigl(8x^2 h_x-\frac{v^2}{4x^2}\Bigr) du \right).
 \end{gather*}
Then we have:
 $$
\alpha_1 \wedge \alpha_2 \wedge \alpha_3 \wedge \alpha_4= (I_0 I_1)^{-1/2} dx \wedge dy\wedge du\wedge dv.
  $$
Metric (\ref{KW}) written in terms of this coframe has coefficients $g_{ij}=g(\nabla_i,\nabla_j)$ and
therefore has the form
 \[
g=I_0  \alpha_1^2+I_1  \alpha_2^2+8 (I_1 I_{2d})^{-1}  \alpha_3^2 - \alpha_3  \alpha_4.
 \]
This suggests that $\nabla_i$ and $I_0,I_1,I_{2d}$ generate the field of differential invariants.
This is indeed true, and can be demonstrated as follows.

The differential invariants appearing as nonzero coefficients in the commutation relations $[\nabla_i, \nabla_j]=K_{ij}^k \nabla_k$ are given by
\com{
 \begin{gather*}
K_{12}^2 = (I_0 I_{2a}-I_{2b}),\
K_{13}^3 = - (I_0 \nabla_3(I_{2b}) +2),\
K_{13}^4 =-\frac{8I_0 I_{2a}}{I_1 I_{2d}},\\
K_{14}^4 = I_0 \nabla_3(I_{2b}),\
K_{23}^2 = - \frac{\nabla_3(I_1)}{2I_1},\
K_{23}^3 = (I_1-I_{2e}) I_{2c}-I_0 I_1 \nabla_3(I_{2c}),\\
K_{23}^4 = \frac{-8I_{2b}}{I_1I_{2d}},
K_{24}^4 = -K_{23}^3,
K_{34}^3 = -1,
K_{34}^4 = \frac{I_{2e}}{2I_0I_1} -\frac{I_1 I_{2d}}{2} \nabla_3\Bigl(\frac1{I_1 I_{2d}}\Bigr).
 \end{gather*}
}
 \begin{gather*}
K_{12}^2 = (I_0 I_{2a}-I_{2b}),\
K_{13}^3 = - (I_0 \nabla_3(I_{2b}) +2),\
K_{13}^4 =-\frac{8I_0 I_{2a}}{I_1 I_{2d}},\\
K_{23}^2 = - \frac{\nabla_3(I_1)}{2I_1},\,
K_{23}^3 = I_{2c}(I_1-I_{2e})-I_0 I_1 \nabla_3(I_{2c}) = -K_{24}^4,\,
K_{34}^3 = -1,\\
K_{14}^4 = I_0 \nabla_3(I_{2b}),\
K_{23}^4 = -\frac{8I_{2b}}{I_1I_{2d}},\
K_{34}^4 = \frac{I_{2e}}{2I_0I_1} -\frac{I_1 I_{2d}}{2} \nabla_3\Bigl(\frac1{I_1 I_{2d}}\Bigr).
 \end{gather*}
In particular we can get the differential invariants $I_{2a},I_{2b},I_{2c},I_{2e}$ from $K_{13}^4, \nabla_1(I_1), \nabla_2(I_1), \nabla_3(I_1)$ thereby verifying that $I_0,I_1,I_{2d}$ are in fact sufficient to be a generating set of differential invariants.

 \begin{remark}
For the $\Sym_0$-action, the invariant derivations  $D_x+\frac{2v}{x} D_v$ and $D_y$ should be used instead of
$\nabla_1,\nabla_2$ (they are not invariant under the reflections). In this case only one coefficient of $g$
is nonconstant, suggesting that one differential invariant and four invariant derivations are sufficient for
generating the field of differential invariants.
 \end{remark}

\subsection{Syzygies}\label{S2.2}

Differential relations among the generators of the algebra of differential invariants are called
differential syzygies. They enter the quotient equation, describing the equivalence classes $\E_\infty/\Sym$.

To simplify notations let us rename the generators $a=I_0,b=I_1,c=I_1 I_{2d}$ and use the iterated derivatives
$f_{i_1...i_r} = (\nabla_{i_r} \circ \cdots \circ \nabla_{i_1})(f)$ for $f=a,b,c$.
We can generate all differential invariants of order $k$ by using only these and $\nabla_1^{k-2} (K_{13}^4)$.
The syzygies coming from the commutation relations of $\nabla_i$ have been described in the previous section.
Thus it is sufficient to only consider iterated derivatives that satisfy $i_1 \leq \cdots \leq i_r$.

These are generated by some simple syzygies
\begin{align*}
a_1=2 a, \quad a_2=0, \quad a_3=0, \quad
a_4=0, \quad b_4=0, \quad c_4=-2 c
\end{align*}
and by two more complicated syzygies that involve differentiation of $b,c$ with respect to
$\nabla_1$, $\nabla_2$, $\nabla_3$ up to order three:
\begin{align*}
0 = &2a^2c^2(2b^2b_{3}b_{233}-2b^2b_{23}b_{33}-3bb_{3}^2b_{23}+3b_{2}b_{3}^3)
-ab(4b^2b_{3}cc_{13}\\ -&4b^2b_{3}cc_{23}-4b^2b_{3}c_{1}c_{3}+4b^2b_{3}c_{2}c_{3}+8b^2b_{33}c^2 -4b^2b_{33}cc_{1} \\+&4b^2b_{33}cc_{2}-2bb_{1}b_{33}c^2-4bb_{3}^2c^2+2bb_{3}^2cc_{1}-4bb_{3}^2cc_{2} +2bb_{3}b_{13}c^2 \\
+&2bb_{3}b_{23}c^2-b_{1}b_{3}^2c^2-3b_{2}b_{3}^2c^2)
-b^2b_{3}c(4bc-2bc_{1}+2bc_{2}-b_{1}c),
\end{align*}
\begin{align*}
0 = &8ab^2c^2(b_{3}b_{123}-b_{3}b_{223}-b_{13}b_{23}+b_{23}^2) \\
+&4abc^2(b_{2}b_{3}b_{13}-b_{2}b_{3}b_{23}-2b_{3}^2b_{12}+4b_{3}^2b_{22})\\
+&ac^2(4b_{1}b_{2}b_{3}^2-12b_{2}^2b_{3}^2)
+16b^3c^2(b_{23}-b_{13}-b_{3}) \\
+&8b^3c((2 c_{1}-2 c_{2}-c_{11}+2 c_{12}-c_{22})b_{3}+(b_{13}-b_{23})(c_{1}-c_{2})) \\
+&b^3(4b_{3}c_{1}^2-8b_{3}c_{1}c_{2}+4b_{3}c_{2}^2)
+bc^2(b_{1}^2b_{3}+2b_{1}b_{2}b_{3})\\
+&b^2c^2(16b_{1}b_{3}+4b_{1}b_{13}-4b_{1}b_{23}-24b_{2}b_{3}-4b_{3}b_{11}+4b_{3}b_{12})\\
+&b^2c(-8b_{1}b_{3}c_{1}+12b_{1}b_{3}c_{2}+12b_{2}b_{3}c_{1}-12b_{2}b_{3}c_{2}).
\end{align*}



\subsection{Comparing Kundt waves}\label{S2.3}


In order to compare two Kundt waves of the form (\ref{KW}) choose four independent differential invariants
$J_1,...,J_4$ of order $k$ such that $\hat d J_1 \wedge \hat d J_2 \wedge \hat d J_3 \wedge \hat d J_4 \neq 0$,
where $\hat d$ is the horizontal differential defined by $(\hat d f) \circ j^k h = d (f \circ j^k h)$ for a
function $f$ on $\E_k$.
Then rewrite the metric in terms of the obtained invariant coframe, similar to what we did in Section \ref{S2.1}:
 \[
g=G_{ij} \hat d J_i \hat d J_j
 \]
where $G_{ij}$ are differential invariants of order $k+1$. For a given Kundt wave metric $g$
the ten invariants $G_{ij}$, expressed as functions of $J_i$, determine its equivalence class.

In practice one can proceed as follows. Let $\hat \partial_i$ be the horizontal frame dual to the coframe
$\hat{d}J_j$. These are commuting invariant derivations, called Tresse derivatives.
In terms of them $G_{ij}=g(\hat \partial_i,\hat \partial_j)$. Together the 14 functions $(J_a,G_{ij})$
determine a map $\sigma_g:M^4\to\R^{14}$ (for a Zariski dense set of $g$)
whose image, called the signature manifold, is the complete invariant of a generic Kundt wave $g$.

In particular, we can take the four \textit{second-order} differential invariants $I_0,I_1,I_{2d},I_{2e}$ that are independent for generic Kundt waves. Then $G_{ij}$ are differential invariants of third order, implying that third order differential invariants are sufficient for classifying generic Kundt waves.

 \begin{remark}
The four-dimensional submanifold $\sigma_g(M^4)\subset\R^{14}$ is not arbitrary.
Indeed, the differential syzygies of the generators $(J_a,G_{ij})$ can be interpreted as a system
of PDE (the quotient equation) with independent $J_a$ and dependent $G_{ij}$.
The signature manifolds, encoding the equivalence classes of Kundt waves, are solutions to this system.
 \end{remark}

\subsection{Example}\label{S2.4}

Consider the class of Kundt waves parametrized by two functions of two variables:
  \begin{equation}\label{Skea+}
h=E(u)-\tfrac14\,\mathcal{S}\bigl(F(u)\bigr)x+F''(u)^2(x^3\pm y),
 \end{equation}
where $\mathcal{S}(F)=\frac{F'''}{F'}-\frac32\bigl(\frac{F''}{F'}\bigr)^2$ is the Schwartz derivative.
This class is $\Sym$-invariant and using the action \eqref{e2}-\eqref{e3} the pseudogroup is
almost fully normalized in passing from this class to
 \begin{equation}\label{Skea}
h(x,y,u)= A(u)+x^3+y.
 \end{equation}
The metric $g$ corresponding to this $h$ was found by Skea in \cite{S} as an example of class of
spacetimes whose invariant classification requires the fifth covariant derivative of the Riemann tensor
(so up to order seven in the metric coefficients $g_{ij}$ equivalently given by $j^7h$).
However with our approach they can be classified via third order differential invariants, and
we will demonstrate how to do it for this simple example.

The transformations from $\Sym_0$ preserving \eqref{Skea} form the two-dimensional non-connected
group $\Sym_0'$: $(x,y,u,A)\mapsto (x,y+c,\pm u+b,A-c)$, and those
of $\Sym$ form the group $\Sym'$ extending $\Sym_0'$ by the map $(x,y,u,A)\mapsto(-x,-y,u,-A)$.
Distinguishing the Kundt waves given by \eqref{Skea+} with respect to pseudogroup $\Sym$ (or $\Sym_0$)
is equivalent to distinguishing the Kundt waves given by \eqref{Skea} with respect to group $\Sym'$
(or $\Sym_0'$).

The differential invariants from Section \ref{S2.1} can be used for this purpose. However the normalization
of \eqref{Skea+} to \eqref{Skea} allows for a reduction from 4-dimensional signature manifolds to signature curves
as follows. The metrics with $A_{uu} \equiv 0$ are easy to classify, so assume $A_{uu} \neq 0$.

The invariants $\sqrt{I_0}=x$, $\sqrt{I_1}=\frac{x h_x-h}{h_y}$, $I_{2d}$, $I_{2e}$ are basic
for the action of $\Sym_0$, and their combination gives simpler invariants $J_1=x$, $J_2=A+y$,
$J_3=v^2$, $J_4= A_u/v$ with
$\frac{\hat dJ_1 \wedge \hat d J_2 \wedge \hat d J_3 \wedge \hat d J_4}
{dx \wedge dy \wedge du \wedge dv}= -2A_{uu}$.
The nonzero coefficients $G_{ij}$ are given by
 \begin{gather*}
G_{11}=1=G_{22}, \quad G_{13}=\frac{J_4}{2 J_1 A_{uu}},\quad G_{14} = \frac{J_3}{J_1 A_{uu}},
\quad G_{23} = -\frac{J_4^2}{2A_{uu}}, \\
G_{33} = -\frac{J_4  (32  J_1^6  J_4-4  J_1^2  J_3  J_4^3+32  J_1^3  J_2  J_4+4  J_1^2  A_{uu}-J_3  J_4)}{16 J_3  A_{uu}^2  J_1^2},\\
G_{34} = \frac{-32 J_1^6 J_4-32 J_4 J_2 J_1^3+(4 J_3 J_4^3-2 A_{uu}) J_1^2+J_4 J_3}{8A_{uu}^2 J_1^2},\\
G_{24} = -\frac{J_3 J_4}{A_{uu}}, \quad G_{44} = \frac{-32 J_1^6 J_3+4 J_1^2 J_3^2 J_4^2-32 J_1^3 J_2 J_3+J_3^2}{4 A_{uu}^2 J_1^2}.
 \end{gather*}
There are five functionally independent invariants, and they are expressed by
$J_1$, $J_2$, $J_3$, $J_4$, $A_{uu}$. Restricted to the specific Kundt wave \eqref{Skea},
only four of them are independent yielding one dependence. This can be interpreted as
a relation between the invariants $A_u^2$ and $A_{uu}$, giving a curve in the plane
due to constraints $A_x=A_y=A_v=0$, and completely determining
the equivalence class. In addition, $A+y$ is a $\Sym_0$-invariant of order 0.

Consequently, two Skea metrics given by \eqref{Skea} are $\Sym_0$-equivalent if their signatures
$\{(A_u(u)^2, A_{uu}(u))\} \subset \mathbb R^2$ coincide as unparametrized curves.
Indeed, let $A_{uu}=f(A_u^2)$ be a signature curve (no restrictions but, for simplicity, we consider
the one that projects injectively to the first components). Viewed as an ODE on $A=A(u)$ it has a
solution uniquely given by the initial data $(A(0),A_u(0))$. This can be arbitrarily changed using
the freedom $(u,y)\mapsto(u+b,y+c)$ of $\Sym_0'$ whence the data encoding $g$ is restored uniquely.

For the $\Sym$-action, we combine the invariants $I_0, I_1 I_{2a}, I_{2d}, I_{2e}$ to construct a
simpler base $J_1=x^2, J_2=(A+y) x, J_3=v^2, J_4=x A_u/v$ of invariants. In this case
we again get $\frac{\hat dJ_1 \wedge \hat d J_2 \wedge \hat d J_3 \wedge \hat d J_4}
{dx \wedge dy \wedge du \wedge dv}= -4x^3 A_{uu}\neq0$, and
basic order 0, 1 and 2 differential invariants for the dimension reduction are
$(A+y)^2$, $A_u^2$, $A_{uu}/(A+y)$.
Proceeding as before we obtain a signature curve $\{(A_u(u)^2, A_{uu}(u)^2)\} \subset \mathbb R^2$ that,
as an unparametrized curve, is a complete $\Sym$-invariant of the Kundt waves of Skea type \eqref{Skea}.

\section{Specification to the vacuum case}\label{S3}

It was argued in Section \ref{S1.1} that the Lie pseudogroup preserving vacuum Kundt waves of the form (\ref{KW})
is the same as the one preserving general Kundt waves of the same form. The PDE
$\ric_k = \{h_{xx}+h_{yy}=0\}^{(k-2)}\cup\E_k$ defining vacuum Kundt waves contains some orbits in
$\E_k$ of maximal dimension. This follows from the proof of Theorem \ref{counting}, since the point
$\theta_k \in \E_k$ chosen there belongs also to $\ric_k$.

This implies that orbits in general position in $\ric_k$ are also orbits in general position in $\E_k$.
Generic vacuum Kundt waves are separated by the invariants
found in Section \ref{S2}, and all previous results are easily adapted to the vacuum case.

\subsection{Hilbert and Poincar\'e function}\label{S3.1}

For vacuum Kundt waves we have additional $\binom{k+1}{3}$ independent differential equations of order $k$ defining $\ric_k\subset\E_k$, so the dimension of $\ric_k$ is $4+(k+1)^2$ for $k\geq 0$.
The codimension of orbits in general position in $\ric_k$ is thus given by
 $$
s^\ric_0=1\ \text{ and }\ s_k^\ric= k(k+1) \text{ for }k\geq1.
 $$
Consequently the Hilbert function $H_k^\ric = s_k^\ric-s_{k-1}^\ric$ is given by
 $$
H_0^\ric=H_1^\ric=1\ \text{ and }\ H_k^\ric =  2k \text{ for } k\geq2.
 $$

The corresponding Poincar\'e function $P_\ric(z)=\sum_{k=0}^\infty H_k^\ric z^k$ is equal to
 \[
P_\ric(z)= \frac{1-z+3z^2-z^3}{(1-z)^2}.
 \]

\subsection{Differential invariants}\label{S3.2}

The differential invariants of second order from Section \ref{S2.1} are still differential invariants in the vacuum case.
The only difference is that two second order invariants $I_{2a},I_{2c}$ become dependent since the vacuum condition implies
$I_{2a}+I_{2c}=0$; in higher order we add differential corollaries of this relation.
It follows that we can generate all $\Sym$-invariants of higher order by using the
differential invariants $I_0, I_1, I_{2d}$ and invariant derivations $\nabla_i$ above.

The differential syzygies found in Section \ref{S2.2} will still hold, but we get some new ones obtained by $\nabla_i$
differentiations of the Ricci-flat condition $I_{2a}+I_{2c}=0$. In terms of the differential invariants
$a,b,c,K_{13}^4$ from Section \ref{S2.2}, the syzygy on $\ric_2$ takes the form
\begin{equation*}
K_{13}^4 b c (a + b)+4 a (2b+ b_{1}+ b_{2}) = 0.
\end{equation*}

The case of $\Sym_0$-invariants is treated similarly.

\subsection{Comparing vacuum Kundt waves} \label{S3.3}

For the basis of differential invariants we can take the same second-order invariants
as for the general Kundt waves: $I_0,I_1,I_{2d},I_{2e}$. Then we express the metric coefficients
$G_{ij}$ in terms of this basis of invariants.

The corresponding four-dimensional signature manifold $\sigma_g(M^4)$ is
restricted by differential syzygies of the general case plus the vacuum constraint.
Considered as an unparametrized submanifold in $\R^{14}$ it completely classifies the metric $g$.

\section{The Cartan-Karlhede algorithm}\label{S4}

Next, we would like to compare the Lie-Tresse approach to differential invariants with Cartan's equivalence method.
We outline the Cartan-Karlhede algorithm for finding differential invariants. The general description of the algorithm
can be found in \cite{Kar}. Its application to vacuum Kundt waves has been recently treated in \cite{MMC}.

\subsection{The algorithm for vacuum Kundt waves}\label{S4.1}

Consider the following null-coframe
in which metric \eqref{KW} has the form $g=2 m \odot \bar{m}-2\ell \odot n$
(as before $h_v=0=h_{xx}+h_{yy}$):
 \[
\ell = du,\quad
n= \frac{1}{2}dv-\frac{v}{x} dx +\left(4xh-\frac{v^2}{8x^2}\right) du,\quad
\begin{array}{l}m=\frac{1}{\sqrt{2}\mathstrut} (dx+i dy),\\ \bar m=\frac{1\mathstrut }{\sqrt{2}}(dx-i dy).\end{array}
 \]
Let $\Delta, D, \delta, \bar\delta$ be the frame dual to coframe $\ell,n,m,\bar m$:
 \[
\Delta = \p_u-\left(8xh-\frac{v^2}{4x^2}\right)\,\p_v,\quad
D= 2\p_v,\quad
\begin{array}{l}\delta=\frac{1}{\sqrt{2}\mathstrut} (\p_x-i\p_y)+\frac{v\sqrt{2}}{x}\,\p_v,\\
\bar\delta=\frac{1}{\sqrt{2}\mathstrut} (\p_x+i\p_y)+\frac{v\sqrt{2}}{x}\,\p_v.\end{array}
 \]

There is a freedom in choosing the (co)frame, encoded as the Cartan bundle.
The general orthonormal frame bundle $\tilde\rho:\tilde{\mathcal P}\to M$ is
a principal bundle with the structure group $O(1,3)$. For Kundt waves the non-twisting
non-expanding shear-free null congruence $\ell$ is up to scale unique, and this reduces
the structure group to the stabilizer $H \subset O(1,3)$ of the line direction $\R\cdot\ell$,
yielding the reduced frame bundle $\rho:\mathcal P\to M$, which is a principal
$H$-subbundle of $\tilde{\mathcal{P}}$.

This so-called parabolic subgroup $H$ has dimension four and the $H$-action on our null (co)frame is given by boosts $(\ell,n)\mapsto (B \ell, B^{-1} n)$,
spins $m \mapsto e^{i\theta} m$ and null rotations $(n,m)\mapsto (n+cm+\bar c \bar m+|c|^2 \ell,m+\bar c \ell)$ about $\ell$, where parameters $B,\theta$ are real and the parameter $c$ is complex.

Let $\nabla$ denote the Levi-Civita connection of $g$, and let $R$ be the Riemann curvature tensor. Written in terms
of the frame, the components of $R$ and its covariant derivatives are invariant functions on $\mathcal P$, but they
are not invariants on $M$. The structure group $H$ acts on them and their $H$-invariant combinations are absolute
differential invariants.

In practice $H$ is used to set as many components of $\nabla^k R$ as possible to constants, as this is a coordinate independent condition for the parameters of $H$.
In the Newman-Penrose formalism \cite{PR}, the Ricci ($\Phi$) and Weyl ($\Psi$) spinors for the Kundt waves are given by
 \[
\Phi_{22}=2x (h_{xx}+h_{yy}), \qquad \Psi_{4} = 2x(h_{xx}-h_{yy}-2i h_{xy}).
 \]
A boost and spin transform $\Psi_4$ to $B^{-2} e^{-2i\theta} \Psi_4$. Thus if $\Psi_4\neq0$ it can be made equal to $1$
by choosing $B^2=4x\sqrt{h_{xx}^2+h_{xy}^2}$ and $e^{2i\theta}= \frac{h_{xx}-i h_{xy}}{\sqrt{h_{xx}^2+h_{xy}^2}}$.

This reduces the frame bundle and the new structure group $H$ is two-dimensional.
In the next step of the Cartan-Karlhede algorithm we use the null-rotations to normalize components of the first covariant derivative of the Weyl spinor.
The benefit of setting $\Psi_4=1$ is that components of the Weyl spinor and its covariant derivatives can be written in terms of the spin-coefficients
and their derivatives. For example, the nonzero components of the first derivative of the Weyl spinor are
 \[
(D \Psi)_{50} = 4 \alpha, \quad (D \Psi)_{51} = 4 \gamma, \quad (D \Psi)_{41} = \tau .
 \]
The null-rotations, with complex parameter $c$, sends $\gamma$ to $\gamma+c \alpha+\frac{5}{4} \bar c \tau$, but leaves
$\alpha$ and $\tau$ unchanged. Assuming that $|\alpha| \neq \frac{5}{4} |\tau|$ it is possible to set $\gamma=0$,
and this fixes the frame. In this case there will be four Cartan invariants of first order in curvature components,
namely the real and imaginary parts of $\alpha$ and $\tau$. They can be expressed in terms of differential invariants as follows:
 \begin{align*}
\alpha &= \frac{-\sqrt{2i}}{8\sqrt{I_0}} \frac{J_-^{1/4}}{J_+^{5/4}}
\left(i\sqrt{I_0I_1}(2I_0I_{2a}^2-I_{2a}+2\nabla_1I_{2a})+2I_{2b}^2-3I_{2b}+2\nabla_1I_{2b}\right)\\
\tau &=  \frac1{\sqrt{2iI_0}} \frac{J_+^{1/4}}{J_-^{1/4}}, \qquad\text{ where }\quad
J_\pm=I_{2b}\pm i\sqrt{I_0I_1}I_{2a}.
 \end{align*}

These give four independent invariant functions on $\mathcal R_\infty$, but when restricted to a vacuum Kundt wave metric
(to the section $j^\infty_Mg\subset \mathcal R_\infty$) at most three of them are independent:
 \[
\hat d (\alpha+\bar \alpha) \wedge \hat d(\alpha-\bar \alpha) \wedge \hat d (\tau+\bar \tau) \wedge \hat d (\tau-\bar \tau)=0.
 \]
The generic stratum of this case corresponds to the invariant branch (0,3,4,4) of the Cartan-Karlhede algorithm in \cite{MMC}.

At the next step of this algorithm the derivatives of the three Cartan invariants from the last step are computed,
resulting in the invariants $\Delta|\tau|, \bar\delta\alpha, \mu, \nu$ (the latter again complex-valued).
One more derivative gives the invariant $\Delta(\Delta|\tau|)$ as a component of the third covariant derivative of the
curvature tensor. Further invariants (when restricted to $j^\infty_Mg$)
will depend on those already constructed, so only 12 real-valued Cartan invariants are required to classify vacuum Kundt waves.
\begin{remark}
	In Section \ref{S2.3} it was stated that 14 differential invariants ($J_a, G_{ij}$) are sufficient for classifying
Kundt waves, but choosing $J_1=I_0,J_2=I_1,J_3=I_{2d},J_4=I_{2e}$ it turns out that we get precisely 12 functionally
independent differential invariants among them.
\end{remark}

\subsection{Cartan invariants vs.\ absolute differential invariants}\label{S4.2}

Let us take a closer look at the relationship between the Cartan invariants and
the differential invariants from Section \ref{S2}.

Differential invariants are functions on $J^\infty \pi$, or on a PDE therein, which are constant on orbits
of the Lie pseudogroup $\mathcal G$. Cartan invariants, on the other hand, are components of the curvature tensor and
its covariant derivatives. These components are dependent on the point in $M$ and the frame.

If we normalize the group parameters and hence fix the frame, i.e.\ a section of the Cartan bundle, then
the Cartan invariants restricted to this section are invariant functions on $J^\infty \pi$.
The following commutative diagram explains the situation.

\begin{center}
\begin{tikzcd}
 &\mathcal{P} \arrow[d, "\rho"] & \arrow[l]\arrow[d] \pi_\infty^* \mathcal P &	\\
&M  &\arrow[l, "\pi_\infty"] \E_\infty  \subset J^\infty \pi
\end{tikzcd}
\end{center}

Initially the Cartan invariants are functions on
 \[
\pi_\infty^* \mathcal P = \{(\omega,g_\infty) \in \mathcal P \times \E_\infty \mid \rho(\omega)=\pi_\infty(g_\infty) \}
 \]
and they suffice to solve the equivalence problem because $\mathcal P$ is equipped with an absolute parallelism $\Omega$
(Cartan connection) whose structure functions generate all invariants on the Cartan bundle.
Indeed, an equivalence of two Lorentzian spaces $(M_1,g_1)$ and $(M_2,g_2)$ lifts to an equivalence between
$(\mathcal P_1,\Omega_1)$ and $(\mathcal P_2,\Omega_2)$ and vise versa
the equivalence upstairs projects to an equivalence downstairs.

Projecting the algebra of invariants on the Cartan bundle to the base we obtain the algebra of
absolute differential invariants consisting of $\mathcal G$-invariant functions on $\E_\infty$.
This is achieved 
by invariantization of the invariants on $\mathcal P$ with respect to the structure group.

This is done in steps by normalizing the group parameters, effecting in further reduction of the structure group.
When the frame is fully normalized (or normalized to a group acting trivially on invariants)
the Cartan bundle is reduced to a section of $\mathcal P$,
restriction to which of the $\nabla^k R$ components gives scalar differential invariants on $M$.
Often these functions and their algebraic combinations that are absolute differential invariants,
evaluated on the metric, are called Cartan invariants.

\subsection{A comparison of the two methods}\label{S4.3}

The definite advantage of Cartan's invariants is their universality. A basic set of invariants can be chosen for almost
the entire class of metrics simultaneously. The syzygies are also fully determined by the commutator relations, the Bianchi and Ricci identities
in the Newman-Penrose formalism \cite{PR}. Yet this basic set is large and algebraically dependent invariants
should be removed, resulting in splitting of the class into different branches of the Cartan-Karlhede algorithm.
See the invariant-count tree for the class of vacuum Kundt waves in \cite{MMC}.

The normalization of group parameters however usually introduces algebraic extensions into the algebra of invariants.
The underlying assumption at the first normalization step in Section \ref{S4.1} is that $\Psi_4$ is nonzero.
This means that
also for Cartan invariants we must restrict to the complement of a Zariski-closed set in $\E_k$.

Setting $\Psi_4$ to $1$ introduces radicals into the expressions of Cartan invariants. A sufficient care with this
is to be taken in the real domain, because the square root is not everywhere defined and is multi-valued.
At this stage it is the choice of the $\pm$ sign, but the multi-valuedness becomes more restrictive
with further invariants. For instance, the expressions for $\alpha$ and $\tau$ contain radicals of $J_\pm$
depending on $\sqrt{I_0I_1}$.

Recall that even though the invariant $I_0$ and $I_1$ are squares, the extraction of the square root
cannot be made $\Sym$-equivariantly and is related to a choice of domain for the pseudogroup $\Sym_0$.
Changing the sign of $\sqrt{I_0I_1}$ results in interchange $J_-\leftrightarrow J_+$ modifying
the formula for $\alpha$ and $\tau$ (which, as presented, is also subject to some sign choices).
The complex radicals carry more multi-valued issues: choosing branch-cuts and restricting to simply connected domains.

Thus Cartan's invariants computed via the normalization technique are only locally defined. In addition,
the domains where they are defined are not Zariski open, in particular they are not dense.

In contrast, elements of the algebra of rational-polynomial differential invariants described in
Section \ref{S2} are defined almost everywhere, on a Zariski-open dense set. The above radicals are avoidable
because we know from Section \ref{S1.2} that generic Kundt waves, as well as vacuum Kundt waves,
can be separated by rational invariants.

Another aspects of comparison is coordinate independence. The class of metrics \eqref{KW}
is given in specific Kundt coordinates, from which we derived the pseudogroup $\Sym$. Changing the coordinates
does not change the pseudogroup, but only its coordinate expression. In other words, this is equivalent to
a conjugation of $\Sym$ in the 
pseudogroup $\op{Diff}_\text{loc}(M)$.

The Cartan-Karlhede algorithm is manifestly coordinate independent, i.e.\ the invariants are 
computed independently of the form in which a Kundt wave is written. However a normalization of parameters is required
to get a canonical frame. It is a simple integration to derive from this Kundt coordinates. It is also
possible to skip integration with the differential invariants approach as abstractly jets are
coordinate independent objects. This would give an equivalent output.

\section{Conclusion}\label{S5}

In this paper we discussed Kundt waves, a class of metrics that are not distinguished by Weyl's scalar curvature
invariants. We computed the algebra of scalar differential invariants that separate generic metrics in the class
and showed that this algebra is finitely generated in Lie-Tresse sense globally.
These invariants also separate the important sub-class of vacuum Kundt waves.

The latter class of metrics was previously investigated via Cartan's curvature invariants in \cite{MMC}
and we compared the two approaches. In particular, we pointed out that normalization in the Cartan-Karlhede
algorithm leads to multi-valuedness of invariants. Moreover, the obtained Cartan's invariants are local even in
jets-variables (derivatives of the metric components). This leads to restriction of domains of definitions,
which in general may not be even invariant with respect to the equivalence group, see \cite{KL2}.

With the differential invariant approach the signature manifold can be reduced in dimension,
as we saw in Section \ref{S2.4}. For the general class of Kundt waves where $h_v=0$, the $v$-variable can be removed
from consideration and furthermore it is not difficult to remove the $y$-variable too. This dimension reduction
leads to a much simpler setup and the classification algorithm. We left additional independent variables to match
the traditional approach via curvature invariants.

The two considered approaches are not in direct correspondence and each method has its own specifications.
For instance, the invariant-count tree in the Cartan-Karlhede algorithm ideologically has a counter-part
in the Poincar\'e function for the Lie-Tresse approach. However orders of the invariants in the two methods
are not related, obstructing to align the filtrations on the algebras of invariants.

For simplicity in this paper we restricted to generic metrics in the class of Kundt waves. This
manifests in a choice of four functionally independent differential invariants, which is not always possible.
For instance, metrics admitting a Killing vector never admit four independent invariants.
With the Cartan-Karlhede approach this corresponds to invariant branches like (0,1,3,3) ending not with 4,
and for the vacuum case all such possibilities were classified in \cite{MMC}.

With the differential invariants approach we treated metrics specified by explicit inequalities:
$h_y\neq0$, $I_0I_1\neq0$, $\dots$, such that the basic invariants and derivations are defined. It is possible
to restrict to the singular strata, and find the algebra of differential invariants with respect to the
restricted pseudogroup. Thus differential invariants also allow to distinguish more special metrics
in the class of Kundt waves.

To summarize, the classical Lie-Tresse method of differential invariants is a powerful alternative to
the Cartan equivalence method traditionally used in relativity applications.

\bigskip

\textsc{Acknowledgement.}
{\footnotesize
BK is grateful to the BFS/TFS project Pure Mathematics in Norway for funding a seminar on the topic of the current research.
The work of DM was supported by the Research Council of Norway, Toppforsk grant no.\,250367:
Pseudo-Rieman\-nian Geometry and Polynomial Curvature Invariants: Classification, Characterisation and Applications.
ES acknowledges partial support of the same grant and hospitality of the University of Stavanger.}


\end{document}